\documentclass[aps,pra,reprint,superscriptaddress,nofootinbib,nobibnotes,floatfix]{revtex4-2}
\usepackage[T1]{fontenc}
\usepackage{lmodern,microtype}
\usepackage{amsmath,amssymb,amsfonts,mathtools,bm}
\usepackage{amsthm}
\usepackage{booktabs,array,enumitem,graphicx,xcolor}
\usepackage{hyperref}
\usepackage[nameinlink,capitalize]{cleveref}
\hypersetup{colorlinks=true,linkcolor=blue!45!black,citecolor=blue!45!black,urlcolor=blue!45!black,
 pdftitle={Optimal temporal hiding in correlated quantum reference-frame processes},
 pdfauthor={Maxim V. Churilov}}

\newtheorem{theorem}{Theorem}
\newtheorem{lemma}[theorem]{Lemma}
\newtheorem{proposition}[theorem]{Proposition}
\newtheorem{corollary}[theorem]{Corollary}

\newcommand{\Tr}{\operatorname{Tr}}

\newcommand{\Prob}{\mathsf P}
\newcommand{\TV}{d_{\mathrm{TV}}}
\newcommand{\F}{\mathbb F}
\newcommand{\wt}{\operatorname{wt}}
\newcommand{\supp}{\operatorname{supp}}
\newcommand{\drel}{d_{\rm rel}}

\makeatletter
\AtBeginDocument{%
  \@ifpackageloaded{hyperref}{\hypersetup{hidelinks}}{}%
}
\makeatother

\begin{document}
\title{Optimal Temporal Hiding in Correlated Quantum Reference-Frame Processes}
\author{Maxim V. Churilov}
\email{churilovm1305@gmail.com}
\affiliation{Independent Researcher, Orenburg, Russia}
\date{April 27, 2026}

\begin{abstract}
A finite group-valued temporal reference frame defines a correlated random-unitary process rather than a list of independent channels.  We study the fresh-probe architecture in which each time slot has its own carrier input while a classical memory correlates the group labels across slots.  In this setting we identify a regime in which the complete trajectory law is an exact operational coordinate.  If the carrier contains every irreducible representation of the finite group $G$, one ancilla-assisted input resolves all group elements, and for arbitrary laws $\mu,\nu$ on $G^n$ the strategy half-distance equals $d_{\mathrm{TV}}(\mu,\nu)$.  Optimal output-side causal post-processing likewise reduces to classical convolution on $G^n$.  This converts temporal hiding into a $q$-ary coding problem.  Uniform coset laws of an $[n,k]_q$ code are identical on every set of fewer than $d(C^\perp)$ slots and perfectly distinguishable globally.  For $M$ sectors with leakage at most $\eta$ from $t$ selected slots and global decoding error $\epsilon$, we obtain
\[
 (1-\epsilon)\log M\le (n-t)\log q+\eta+h_2(\epsilon).
\]
For nested codes $C\subset D$, the payload $R=\dim D-\dim C$, invisible depth $t=d_{\rm rel}(C^\perp,D^\perp)-1$, and sector distance $d_{\rm rel}(D,C)$ obey
\[
 R+t+d_{\rm rel}(D,C)-1\le n,
 \qquad R+t+2e+f\le n,
\]
where $e$ and $f$ are adversarial errors and known erasures.  Nested generalized Reed--Solomon codes attain the bounds in their existence range.  We also give the exact projection-rank leakage profile and show that hiding arbitrary coherent superpositions from $t$ slots is precisely quantum erasure correction, yielding $\log_q K+2t\le n$.  The results separate classical trajectory privacy from coherent temporal privacy and quantify the robust payload hidden from reduced process tomography.
For an arbitrary \(n\)-history alphabet, the exact coherent distance to
the entire classical simplex is \(1-1/n\).  Under \(m\) repeated uses,
the exact distance from every correlated classical-history model is
\(1-n^{-m}\), giving a fully adaptive discrimination strong converse
with error \(1/(2n^m)\).
\end{abstract}

\maketitle

\section{Introduction}

A reference token used at several times carries a history, not merely a sequence of independent noise parameters.  The hidden orientation, phase, or clock offset may persist and correlate different interrogations.  Quantum combs and process tensors provide the appropriate multi-time operational formalism \cite{GutoskiWatrous2007,Chiribella2009,Pollock2018,MilzModi2021,Keeling2026}, while generalized process divergences clarify the required data-processing properties \cite{Zambon2024}.  Quantum reference frames supply the relational setting \cite{Bartlett2007,Hausmann2025,Carette2025}.  Here we study a deliberately restricted but exact model: a classical finite-group frame history drives a correlated sequence of unitary branches, one on each fresh process input.  There is no quantum system wire carrying the same probe from one slot to the next; the inter-slot memory is the sampled classical history.  The central identification problem is how much of that process can remain invisible to experiments with limited temporal support.

The statement that ``correlations matter'' is too weak.  A useful theory should determine the full operational distance, identify the most general allowed degradation, construct processes hidden up to a prescribed depth, and prove whether the construction is optimal.  We solve these questions for finite group-valued frame histories whose branches can be resolved by one ancilla-assisted input.

Let $\bm g=(g_1,\ldots,g_n)\in G^n$ be sampled from an arbitrary law $\mu$.  At time $j$, the token applies $U_{g_j}$.  No Markov, stationarity, or independence assumption is made.  If the representation contains every irreducible type, one calibration state has an orthonormal group orbit.  Feeding one such state at every time makes complete trajectories orthogonal.  The resulting strategy norm is exactly total variation, and arbitrary output-side causal quantum processing cannot beat a correlated classical correction history.

The second part turns this exact coordinate into a coding problem.
Uniform laws on distinct cosets of a code have disjoint global supports
but common low-coordinate marginals controlled by the dual distance.
Their corresponding QRF processes are therefore locally identical and
globally orthogonal.  An information-theoretic converse remains valid
with imperfect global decoding and nonzero local leakage.  Without
assuming linearity or uniformity, a puncturing argument gives the
privacy--distance Singleton bound
$\log_qM+t+d_Z-1\le n$.  Robust readout needs one further constructive
idea.  All cosets of one code are separated by Hamming distance one,
so their syndrome labels cannot correct even a single adversarial
branch-readout error.  We repair this obstruction with nested pairs
$C\subset D$: cosets of $C$ inside $D$ retain the hiding depth of $C$
and acquire the relative distance of the pair.  Nested generalized
Reed--Solomon codes attain the general bound.

\subsection{Relation to prior work and claim boundary}

Quantum strategies and process tensors provide the ambient formalism \cite{GutoskiWatrous2007,Chiribella2009,Pollock2018,Keeling2026}.  Generalized process divergences and their data processing were analyzed in \cite{Zambon2024}.  Operational transformations between quantum reference-frame descriptions were developed in \cite{Carette2025}; our object is instead a restricted stochastic multi-slot process with a fixed external slot structure.  The equivalence between code dual distance and orthogonal-array strength is classical \cite{Delsarte1973,MacWilliamsSloane1977,HuffmanPless2003}.

We do not claim that $t$-wise independent code distributions,
orthogonal arrays, wiretap-II privacy, or relative distances of nested
codes are new; these are standard in coding and ramp secret sharing
\cite{Delsarte1973,MacWilliamsSloane1977,HuffmanPless2003,OzarowWyner1984,Galindo2018}.
The contribution is their operational realization as finite-history QRF
processes, the strategy-norm and causal-degradation isometries for
orbit-resolving carriers, the finite-error capacity converse, and a
self-contained privacy--distance converse for simultaneous temporal
hiding and robust sector readout.  The projection-rank leakage hierarchy
and the coherent Singleton boundary are standard coding-theoretic
structures \cite{Kurihara2012,Galindo2018,KnillLaflamme2000,Grassl2022};
we use them to identify the exact operational boundary between classical
trajectory privacy and coherent temporal privacy.  The model is
process-separable and does not cover arbitrary coherent superpositions
of histories or indefinite causal order.

\paragraph{Status of the principal ingredients.}
Comb discrimination supplies the operational norm, while dual
distance, orthogonal arrays, and nested-code relative weights supply
the classical construction.  The central package is the exact QRF
process translation, the approximate capacity theorem, the nonlinear
privacy--distance converse, and the corrected
payload--depth--protection constraint.

\section{Orbit-resolving temporal frame processes}

Let $G$ be finite and $U:G\to\mathcal U(\mathcal H)$ a unitary representation.  We call the carrier \emph{orbit resolving} if there are an ancilla $R$ and a unit vector $|\Psi\rangle\in\mathcal H\otimes R$ such that
\begin{equation}
 |\Psi_g\rangle=(U_g\otimes I_R)|\Psi\rangle,
 \qquad g\in G,
 \label{eq:orbit}
\end{equation}
form an orthonormal family.

\begin{lemma}[Orbit-resolution criterion]
\label{lem:orbit-criterion}
A finite-group representation $U$ is orbit resolving if and only if it
contains every irreducible representation of $G$ at least once.  In that
case the seed may be chosen with a reduced carrier state $\rho$ satisfying
$[\rho,U_g]=0$ and $\Tr(\rho U_g)=\delta_{g,e}$.
\end{lemma}

\begin{proof}
Write $U\simeq\bigoplus_\lambda(U^\lambda\otimes I_{m_\lambda})$.  If every
$m_\lambda$ is nonzero, choose one multiplicity vector in each block and set
\[
 \rho=\bigoplus_\lambda \frac{d_\lambda}{|G|}
 I_{d_\lambda}\otimes |0_\lambda\rangle\!\langle0_\lambda|.
\]
By the regular-character identity,
\begin{equation}
 \Tr(\rho U_g)=\frac{1}{|G|}\sum_\lambda
 d_\lambda\chi_\lambda(g)=\delta_{g,e}.
\end{equation}
A purification of $\rho$ therefore has the orthonormal orbit
\eqref{eq:orbit}.  Conversely, orthogonality makes the coefficient function
$g\mapsto\Tr(\rho U_g)$ equal to $\delta_{g,e}$.  The Fourier transform of
$\delta_{g,e}$ is nonzero in every irreducible block, while matrix
coefficients of $U$ vanish in blocks absent from $U$.  Hence no irreducible
type can be missing.
\end{proof}

The regular representation is one example, but multiplicity one for each
irreducible type already suffices after adjoining an ancilla.

At $n$ ordered times, a history law $\mu\in\Prob(G^n)$ defines the correlated unitary process below.  Each slot accepts a fresh carrier input and returns its corresponding output; the tensor-product notation is therefore literal at the level of slot interfaces.  The same sampled history $\bm g$ is stored classically and correlates the branches.  A model in which one quantum probe propagates sequentially through all $U_{g_j}$ is different and is not covered by the trajectory-isometry theorem.
\begin{equation}
 \mathfrak M_\mu^{(n)}
 =\sum_{\bm g\in G^n}\mu(\bm g)
 \bigotimes_{j=1}^n\operatorname{Ad}_{U_{g_j}}.
 \label{eq:history-process}
\end{equation}
The notation is the flattened channel of an $n$-slot comb.  Operationally, a classical environment samples $\bm g$ once, stores it, and applies the indicated branch at each time.  A tester may prepare entangled inputs in advance, retain quantum memory, adapt operations between slots, and measure at the end.  Let $\|\cdot\|_{\rm strat}$ denote the resulting strategy norm, normalized so that perfectly distinguishable processes have half-distance one.

\begin{theorem}[Trajectory isometry]
\label{thm:isometry}
For an orbit-resolving carrier and arbitrary history laws $\mu,\nu$,
\begin{equation}
 \boxed{
 \frac12\|\mathfrak M_\mu^{(n)}-
 \mathfrak M_\nu^{(n)}\|_{\rm strat}
 =\TV(\mu,\nu)}.
 \label{eq:isometry}
\end{equation}
A parallel, nonadaptive tester is optimal.
\end{theorem}

\begin{proof}
Feed one copy of the orbit seed $|\Psi\rangle$ at each time and retain all ancillas.  The branch output
\begin{equation}
 |\Psi_{\bm g}\rangle=\bigotimes_{j=1}^n|\Psi_{g_j}\rangle
\end{equation}
is orthonormal in $\bm g$.  The output difference is diagonal in this basis and has trace norm $\|\mu-\nu\|_1$, giving the lower bound.  For any normalized tester $T$,
\begin{equation}
 T[\mathfrak M_\mu-\mathfrak M_\nu]
 =\sum_{\bm g}(\mu-\nu)(\bm g)\rho_{\bm g}
\end{equation}
for density operators $\rho_{\bm g}$.  The triangle inequality gives the matching upper bound.
\end{proof}

Thus the trajectory law is an operationally complete coordinate of this process family; adaptivity cannot reveal more distance than the orthogonal-history probe already extracts.

\section{Exact output-side causal degradation}

An \emph{output-side causal converter} is a sequence of channels that acts after each process slot on the current output and a retained converter memory, without modifying the input interface of later slots.  It has no access to the tester's retained reference systems.  Such converters include correlated, time-inhomogeneous, and quantum-memory strategies.  Let $\mathsf{OutCausal}_n$ denote this class and define
\begin{equation}
 \delta_{\rm out}(\nu|\mu)=
 \inf_{\Lambda\in\mathsf{OutCausal}_n}
 \frac12\|\mathfrak M_\nu^{(n)}-
 \Lambda\circ\mathfrak M_\mu^{(n)}\|_{\rm strat}.
 \label{eq:out-def}
\end{equation}
On $G^n$, multiplication and inversion are componentwise, and
\begin{equation}
 (r*\mu)(\bm h)=\sum_{\bm g}r(\bm h\bm g^{-1})\mu(\bm g).
 \label{eq:history-conv}
\end{equation}

\begin{theorem}[Causal classicalization]
\label{thm:causal}
For every orbit-resolving carrier,
\begin{equation}
 \boxed{
 \delta_{\rm out}(\nu|\mu)
 =\min_{r\in\Prob(G^n)}\TV(\nu,r*\mu)}.
 \label{eq:causal}
\end{equation}
An optimizer is implemented by sampling the full correction history at the initial time and applying $U_{r_j}$ after slot $j$.
\end{theorem}

\begin{proof}
Any output-side causal converter induces a CPTP map on the tensor product of all output systems, subject to additional causality constraints.  Dropping those constraints can only decrease the infimum.  The product representation of $G^n$ is orbit resolving with seed $|\Psi\rangle^{\otimes n}$.  To see the resulting lower bound directly, let $\rho$ be the reduced carrier state of $|\Psi\rangle$ and let $\Lambda(X)=\sum_aV_aXV_a^\dagger$ be an arbitrary global CPTP map.  Measuring the orthogonal history projectors after branch $\bm g$ gives the subtransition kernel
\begin{equation}
 K(\bm h|\bm g)=\sum_a
 \left|\Tr\!\left(\rho^{\otimes n}
 U_{\bm h}^\dagger V_aU_{\bm g}\right)\right|^2
 =r_0(\bm h\bm g^{-1}),
\end{equation}
where, explicitly,
\begin{equation}
 r_0(\bm x)=\sum_a
 \left|\Tr\!\left(\rho^{\otimes n}U_{\bm x^{-1}}V_a\right)\right|^2.
\end{equation}
The last equality follows by cyclicity and $[\rho,U_g]=0$; the inverse fixes the left-convolution convention for a non-Abelian group.  Its total mass $m\leq1$ is independent of $\bm g$; the missing mass is leakage outside the orbit subspace.  The measured output law is therefore $r_0*\mu$ plus one leakage outcome of probability $1-m$.  For any nonnegative completion $s$ of mass $1-m$, the probability law $r=r_0+s$ satisfies
\[
 \TV(\nu,r*\mu)\leq
 \tfrac12\bigl(\|\nu-r_0*\mu\|_1+1-m\bigr),
\]
which is exactly the total-variation distance after adjoining the leakage outcome to the target with zero weight.  Measurement contractivity and the parallel history tester therefore give the lower bound in \cref{eq:out-def}.

Conversely, sample $\bm r\sim r$ before the first slot and store it classically.  Applying $U_{r_j}$ after slot $j$ maps $U_{g_j}$ to $U_{r_j}U_{g_j}=U_{r_jg_j}$, hence maps $\mu$ to the left convolution $r*\mu$.  This converter is causal, and \cref{thm:isometry} makes its error exactly $\TV(\nu,r*\mu)$.
\end{proof}

The theorem does not include superprocesses that alter future input ports or temporally coarse-grain the process.  That larger transformation class is physically different and can change process distinguishability in ways absent here \cite{Zambon2024}.

\section{Code-coset temporal sectors}

Set $G=(\F_q,+)$, with $q$ a prime power.  Let $C\le\F_q^n$ be a linear $[n,k,d]_q$ code and
\begin{align}
 C^\perp&=\{y:y\cdot x=0\ \text{for all }x\in C\},\nonumber\\
 d^\perp&=\min_{0\ne y\in C^\perp}\wt(y).
\end{align}
For a coset $a+C$, let $\mu_a$ be its uniform law.

\begin{lemma}[Dual distance and local uniformity]
\label{lem:uniformity}
For every coordinate set $S$ with $|S|<d^\perp$, the projection of $C$ onto $\F_q^S$ is surjective with equal fibers.  Hence the $S$-marginal of every $\mu_a$ is the uniform law on $\F_q^S$, independent of the coset.
\end{lemma}

\begin{proof}
If the projection were not surjective, a nonzero linear functional on $\F_q^S$ would annihilate its range.  Extending that functional by zero outside $S$ produces a nonzero dual codeword supported in $S$, contradicting $|S|<d^\perp$.
\end{proof}

\begin{theorem}[Strict temporal-depth hierarchy]
\label{thm:hierarchy}
For distinct cosets $a+C$ and $b+C$:
\begin{enumerate}[label=(\roman*),leftmargin=2.2em]
\item all reduced processes on fewer than $d^\perp$ time slots are identical;
\item every tester whose active interaction support contains fewer than $d^\perp$ slots has zero discrimination advantage;
\item the full processes are perfectly distinguishable:
\begin{equation}
 \frac12\|\mathfrak M_{\mu_a}^{(n)}-
 \mathfrak M_{\mu_b}^{(n)}\|_{\rm strat}=1.
 \label{eq:perfect}
\end{equation}
\end{enumerate}
There are $q^{n-k}$ pairwise orthogonal sectors, and a branch-resolving tester followed by a parity-check matrix reads the sector exactly.
\end{theorem}

\begin{proof}
The first statement is \cref{lem:uniformity}.  A tester supported on a fixed set of slots is a functional of that reduced comb, proving (ii).  Randomized or outcome-dependent choices of at most that many slots also give identical transcripts by conditioning on each decision-tree path.  Distinct cosets are disjoint, so their total-variation distance is one; \cref{thm:isometry} proves (iii).  Measuring the orthogonal branch at each time recovers $\bm g$, and a parity-check matrix $H$ returns the syndrome $H\bm g=Ha$.
\end{proof}

For the parity code
\begin{equation}
 C_{\rm par}=\{x\in\F_q^n:\textstyle\sum_jx_j=0\},
 \label{eq:parity-code}
\end{equation}
we have $d^\perp=n$.  Its $q$ sectors are therefore invisible to every proper temporal marginal and perfectly distinguishable at depth $n$.

\section{Optimal hiding capacity}

The support-counting argument can be strengthened to an operational finite-error statement.  Let a uniformly distributed sector label $J\in\{1,\ldots,M\}$ select a history $X\in\F_q^n$ according to $\mu_J$.  All logarithms and entropies in this section use the same base.  For $S\subseteq[n]$, define the depth-$S$ leakage
\begin{equation}
 L(S)=I(J;X_S).
 \label{eq:leakage}
\end{equation}

\begin{proposition}[Operational meaning of leakage]
\label{prop:leakage}
For an orbit-resolving carrier, $L(S)$ is the maximum mutual information between $J$ and the classical output of a tester restricted to the slots in $S$.
\end{proposition}

\begin{proof}
The branch-resolving seed and orbit measurement recover $X_S$ exactly, so the value in \cref{eq:leakage} is attainable.  Conversely, conditional on $X_S=x_S$, any restricted tester produces a state and then a classical outcome through a channel that is independent of $J$.  The Markov chain $J\to X_S\to Y$ and data processing give $I(J;Y)\le I(J;X_S)$.
\end{proof}

\begin{theorem}[Finite-error temporal hiding capacity]
\label{thm:approx-capacity}
Suppose a decoder $\widehat J(X)$ using the full history has average error probability at most $\epsilon<1$.  If, for some $S\subseteq[n]$ with $|S|=t$,
\begin{equation}
 I(J;X_S)\le\eta,
 \label{eq:eta-leak}
\end{equation}
then
\begin{equation}
 \boxed{
 (1-\epsilon)\log M
 \le (n-t)\log q+\eta+h_2(\epsilon)}.
 \label{eq:approx-capacity}
\end{equation}
Here $h_2$ is binary entropy in the chosen logarithmic base.  The same conclusion holds if \cref{eq:eta-leak} is required for every $t$-set; only one such set is needed for the converse.
\end{theorem}

\begin{proof}
The chain rule and the alphabet size of $X_{S^c}$ give
\begin{align}
 I(J;X)
 &=I(J;X_S)+I(J;X_{S^c}|X_S)\nonumber\\
 &\le \eta+H(X_{S^c})
 \le \eta+(n-t)\log q.
 \label{eq:MI-upper}
\end{align}
Fano's inequality \cite{CoverThomas2006} gives
\begin{equation}
 H(J|X)\le h_2(\epsilon)+\epsilon\log(M-1)
 \le h_2(\epsilon)+\epsilon\log M.
\end{equation}
Substitute these two bounds into $H(J)=I(J;X)+H(J|X)=\log M$ and rearrange.
\end{proof}

\begin{corollary}[From marginal distance to capacity]
\label{cor:tv-capacity}
Suppose that for a fixed $t$-set $S$ there is a law $\nu_S$ such that
\begin{equation}
 \TV((\mu_j)_S,\nu_S)\le\delta
 \qquad(j=1,\ldots,M).
\end{equation}
Writing $Q=q^t$ and assuming $0\le\delta\le1-1/Q$, one may take
\begin{equation}
 \eta=2\bigl[\delta\log(Q-1)+h_2(\delta)\bigr]
 \label{eq:tv-to-info}
\end{equation}
in \cref{eq:approx-capacity}.
\end{corollary}

\begin{proof}
Let $\overline\mu_S=M^{-1}\sum_j(\mu_j)_S$.  Both $(\mu_j)_S$ and $\overline\mu_S$ are within $\delta$ of $\nu_S$.  The sharp finite-alphabet entropy continuity bound \cite{Audenaert2007} applied twice to
\begin{equation}
 I(J;X_S)=H(\overline\mu_S)-\frac1M\sum_jH((\mu_j)_S)
\end{equation}
gives \cref{eq:tv-to-info}.
\end{proof}

\begin{corollary}[Exact capacity and MDS saturation]
\label{thm:capacity}
If the sectors are globally perfectly readable and reveal no information from every set of $t$ fixed slots, then
\begin{equation}
 \boxed{M\le q^{n-t}}.
 \label{eq:capacity}
\end{equation}
Whenever an $[n,t]_q$ MDS code exists, its cosets attain equality and are invisible to all tests of depth at most $t$.
\end{corollary}

\begin{proof}
Set $\eta=0$ and $\epsilon=0$ in \eqref{eq:approx-capacity} to obtain $\log M\leq(n-t)\log q$.  An $[n,t]_q$ MDS code has uniform projections on every $t$ coordinates and disjoint cosets, so its $q^{n-t}$ cosets attain equality.
\end{proof}

The corollary no longer assumes uniform marginals or disjoint supports explicitly: perfect decoding makes the label a function of the full history, and zero leakage is the operational condition.  In the original support formulation one obtains the finer elementary bound
\begin{equation}
 M\,|\supp\nu_S|\le q^n
 \label{eq:support-capacity}
\end{equation}
whenever all sectors have a common $S$-marginal $\nu_S$ and disjoint supports.  Uniformity on $\F_q^t$ recovers \cref{eq:capacity}.  Generalized Reed--Solomon codes attain it in their standard existence range.

\begin{proposition}[Equality structure: large sets of orthogonal arrays]
\label{prop:equality-structure}
Suppose $M=q^{n-t}$ sectors have pairwise disjoint supports and every
$t$-coordinate marginal is uniform on $\F_q^t$.  Then every sector
support has exactly $q^t$ histories, the supports partition
$\F_q^n$, and projection onto any $t$ coordinates is a bijection on
each support.  Equivalently, the sectors form a large set of
index-one orthogonal arrays $\operatorname{OA}(q^t,n,q,t)$.  Linear
MDS cosets are one realization, but linearity is not forced.
\end{proposition}

\begin{proof}
Uniformity on one $t$-set requires at least one support point over each
of its $q^t$ values, so every sector support has size at least $q^t$.
Disjointness and $M=q^{n-t}$ give
\begin{equation}
 q^n\ge\sum_{j=1}^M|\supp\mu_j|
 \ge q^{n-t}q^t=q^n.
\end{equation}
Every inequality is therefore an equality.  Each support has size
$q^t$, the supports cover the history space, and every one of the
$q^t$ projected values has exactly one preimage.  Repeating the
argument for every $t$-set gives the orthogonal-array property.
\end{proof}

\begin{figure}[t]
\centering
\includegraphics[width=\columnwidth]{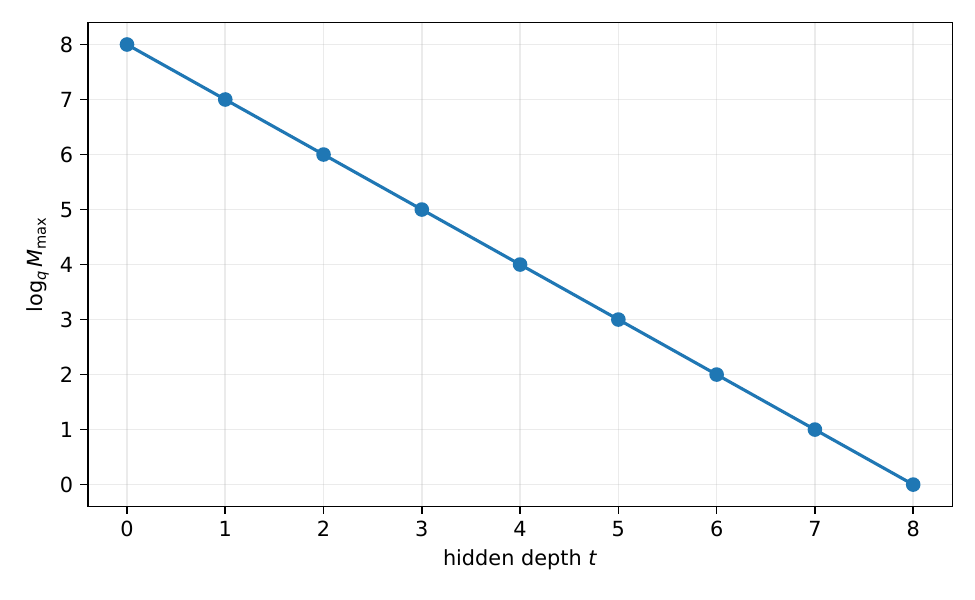}
\caption{The maximum number of globally orthogonal sectors with common uniform marginals on every $t$-slot set is $q^{n-t}$.  MDS-code cosets attain the bound.}
\label{fig:capacity}
\end{figure}

\begin{theorem}[Nonlinear privacy--distance Singleton bound]
\label{thm:nonlinear-singleton}
Let $M\ge2$ history sectors on $\F_q^n$ have pairwise support distance
at least $d_Z$.  If their marginals on one coordinate set $A$ of size
$t$ are identical, then
\begin{equation}
 \boxed{\log_qM+t+d_Z-1\le n}.
 \label{eq:nonlinear-singleton}
\end{equation}
If the sector must be recovered after every pattern of at most $e$
symbol errors and $f$ known erasures, then $d_Z\ge2e+f+1$ and
\begin{equation}
 \boxed{\log_qM+t+2e+f\le n}.
 \label{eq:nonlinear-protection}
\end{equation}
No linearity, uniformity, or additive group structure is required.
\end{theorem}

\begin{proof}
Choose a value $y$ in the support of the common $A$-marginal.  For
each sector $j$, select a history $x_j$ in its support with
$(x_j)_A=y$.  Distinct selected histories have Hamming distance at
least $d_Z$ and agree on $A$.  Consequently $d_Z\le n-t$.  Delete the
$t$ coordinates in $A$ and any further set $B\subseteq A^c$ of
$d_Z-1$ coordinates.  The remaining restrictions of the $x_j$ are
distinct: equality of two restrictions would make the corresponding
histories differ only inside $B$, contradicting their distance.
Hence
\begin{equation}
 M\le q^{n-t-(d_Z-1)},
\end{equation}
which proves \cref{eq:nonlinear-singleton}.  Two sector histories
compatible with the same received word after $e$ errors and $f$
erasures could differ in at most $2e+f$ coordinates.  Thus unique
recovery requires and is guaranteed by $2e+f<d_Z$.
\end{proof}

\section{Robust sector readout requires nested codes}

The primal distance of $C$ separates histories \emph{inside one coset}; it does not separate different cosets.  Indeed, if all cosets of $C$ are used as labels, their pairwise set distance is one.  A single symbol error can therefore change the syndrome to another valid sector.  Simultaneous hiding and adversarial readout protection requires a nested pair.

Let
\begin{equation}
 \begin{gathered}
 C\subset D\subseteq\F_q^n,\qquad \dim C=k,\quad \dim D=K,\\
 R=K-k.
 \end{gathered}
 \label{eq:nested}
\end{equation}
and use the $q^R$ cosets $a+C$ with $a\in D/C$ as sectors.  Define their relative distance
\begin{equation}
 \drel(D,C)=\min\{\wt(x):x\in D\setminus C\}.
 \label{eq:relative-distance}
\end{equation}

\begin{theorem}[Nested temporal sectors]
\label{thm:nested}
The nested construction \cref{eq:nested} has the following properties.
\begin{enumerate}[label=(\roman*),leftmargin=2.2em]
\item Its payload is $R$ $q$-ary symbols, i.e. there are $q^R$ sectors.
\item Every set of at most
\begin{equation}
 t=\drel(C^\perp,D^\perp)-1
 \label{eq:nested-depth}
\end{equation}
slots has the same marginal in every sector.  The common marginal is uniform on the projected subspace $P_S(C)=P_S(D)$; it is uniform on all of $\F_q^S$ whenever $P_S(C)=\F_q^S$.
\item The minimum Hamming distance between two distinct sector supports
is $\drel(D,C)$.  Hence the sector is recovered uniquely after any
pattern of $e$ errors and $f$ known erasures satisfying
\begin{equation}
 2e+f\le\drel(D,C)-1.
 \label{eq:nested-errors}
\end{equation}
\item The parameters obey the relative Singleton bound
\begin{equation}
 \boxed{R+t+\drel(D,C)-1\le n},
 \label{eq:relative-singleton}
\end{equation}
and consequently
\begin{equation}
 \boxed{R+t+2e+f\le n}.
 \label{eq:payload-tradeoff}
\end{equation}
\end{enumerate}
\end{theorem}

\begin{proof}
The quotient $D/C$ has dimension $R$, proving (i).  For a coordinate set $S$, the $S$-marginal of the uniform law on $a+C$ is uniform on the affine space $P_S(a)+P_S(C)$.  These affine spaces agree for every $a\in D$ exactly when $P_S(D)=P_S(C)$.  If equality failed, a linear functional supported on $S$ would annihilate $C$ but not $D$, producing an element of $C^\perp\setminus D^\perp$ of weight at most $|S|$.  This proves (ii).  For distinct $a+C,b+C$ inside $D$,
\begin{equation}
 \min_{c,c'\in C}\wt[(a+c)-(b+c')]
 =\min_{c\in C}\wt(a-b+c).
\end{equation}
The difference lies in $D\setminus C$, and minimizing over all
distinct quotient classes gives \cref{eq:relative-distance}.  The
standard error--erasure distance criterion proves (iii).

For (iv), abbreviate
$d_X=\drel(C^\perp,D^\perp)=t+1$ and
$d_Z=\drel(D,C)$.  Choose a coordinate set $A$ with
$|A|=d_X-1$.  Since $P_A(D)=P_A(C)$, every class in $D/C$ has a
representative that vanishes on $A$.  Choose such representatives
linearly by cleaning a basis of the quotient and extending linearly.
If $|A^c|<d_Z$, any nonzero class would thereby acquire a
representative of weight below $d_Z$, a contradiction.  We may
therefore choose $B\subseteq A^c$ with $|B|=d_Z-1$.  Restrict the
cleaned representative to $Q=[n]\setminus(A\cup B)$.  This defines an
injective linear map $D/C\to\F_q^Q$: if two classes had the same
restriction, the difference of their cleaned representatives would be
a word of $D\setminus C$ supported in $B$, contradicting the
definition of $d_Z$.  Hence
\begin{equation}
 R\le |Q|=n-(d_X-1)-(d_Z-1),
\end{equation}
which is \cref{eq:relative-singleton}.  The error--erasure bound
follows from $2e+f\le d_Z-1$.
\end{proof}

\begin{corollary}[Nested MDS saturation]
\label{cor:nested-mds}
Let $C\subset D$ be nested generalized Reed--Solomon codes of dimensions $k<K$ on the same $n$ evaluation points.  Then
\begin{equation}
 t=k,
 \qquad R=K-k,
 \qquad \drel=n-K+1,
 \label{eq:mds-nested-parameters}
\end{equation}
and equality holds in \cref{eq:relative-singleton}.  Every allocation
of errors and erasures with $2e+f=n-K$ also saturates
\cref{eq:payload-tradeoff}.
\end{corollary}

\begin{proof}
For nested generalized Reed--Solomon codes, puncturing and duality give $\drel(C^\perp,D^\perp)=k+1$ and $\drel(D,C)=n-K+1$.  Substitution into \Cref{thm:nested} gives \eqref{eq:mds-nested-parameters} and equality in the two tradeoffs.
\end{proof}

For example, nested $[6,2]_7\subset[6,4]_7$ Reed--Solomon codes give
$R=2$, $t=2$, and $\drel=3$.  They produce $49$ temporal sectors,
hide the label from every two-slot tester, and correct either one
arbitrary branch-readout error or two known erasures.  Both
$2+2+2=6$ and $2+2+0+2=6$ saturate
\cref{eq:payload-tradeoff}.  By contrast, the binary parity
construction uses $C=[n,n-1,2]_2$ and $D=\F_2^n$: it has $R=1$,
$t=n-1$, and $\drel=1$, saturating the bound with $e=f=0$.

\begin{figure}[t]
\centering
\includegraphics[width=\columnwidth]{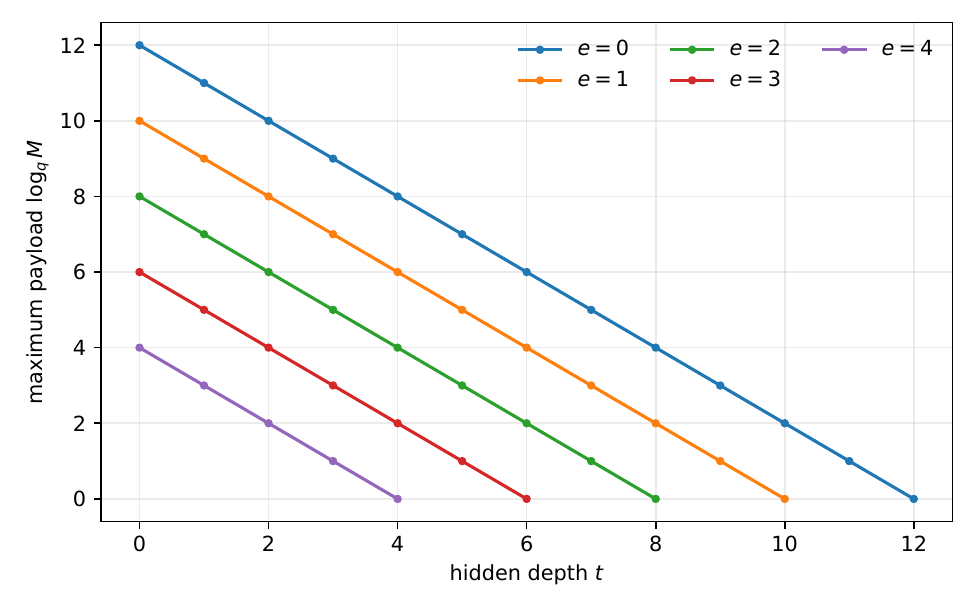}
\caption{The zero-erasure slice of the exact outer region
$\log_qM+t+2e+f\le n$ for $n=12$.  Integer points on the boundary are
attained by suitable nested MDS pairs whenever the required field size
and parity permit.}
\label{fig:tradeoff}
\end{figure}

\section{Exact leakage profile of a nested temporal code}

The zero-leakage depth records only the first point at which a subsystem can
see the quotient label.  The complete leakage curve is the familiar rank
profile of a linear ramp secret-sharing scheme \cite{Kurihara2012,Galindo2018};
here it becomes an operational tester-leakage formula.

\begin{theorem}[Projection-rank leakage formula]
\label{thm:nested-leakage-profile}
Let $C\subset D\subset\mathbb F_q^n$.  Choose $J$ uniformly in $D/C$ and,
conditional on $J=a+C$, choose $X$ uniformly in the coset $a+C$.  For every
coordinate set $S\subseteq[n]$,
\begin{equation}
 \boxed{\qquad
 I(J;X_S)=
 \bigl(\dim P_S(D)-\dim P_S(C)\bigr)\log q.
 \qquad}
 \label{eq:nested-leakage-rank}
\end{equation}
Consequently the operational leakage available to an orbit-resolving tester
on $S$ is an integer multiple of $\log q$.  Define
\begin{equation}
 \begin{aligned}
 d_\ell^{\rm leak}(C,D)
 &=\min\bigl\{|S|:\dim P_S(D)-\dim P_S(C)\geq\ell\bigr\},\\
 &\hspace{8em}1\leq\ell\leq R.
 \end{aligned}
 \label{eq:leakage-weight-hierarchy}
\end{equation}
Then $d_\ell^{\rm leak}$ is the least number of temporal slots needed to
reveal at least $\ell$ $q$-ary symbols of the hidden sector.  In particular,
$d_1^{\rm leak}=d_{\rm rel}(C^\perp,D^\perp)$ and the invisible depth in
Theorem~\ref{thm:nested} is $d_1^{\rm leak}-1$.
\end{theorem}

\begin{proof}
The unconditional history $X$ is uniform on $D$, so its projection is
uniform on the linear space $P_S(D)$ and
$H(X_S)=\dim P_S(D)\log q$.  Conditional on any quotient class, $X_S$ is
uniform on an affine translate of $P_S(C)$, hence
$H(X_S\mid J)=\dim P_S(C)\log q$.  Subtraction proves
\eqref{eq:nested-leakage-rank}.  The operational statement follows from
Proposition~\ref{prop:leakage}.  Finally,
$P_S(D)\neq P_S(C)$ exactly when a linear functional supported in $S$
annihilates $C$ but not $D$, equivalently when
$C^\perp\setminus D^\perp$ contains a nonzero word supported in $S$.
This is the relative-distance characterization used in
Theorem~\ref{thm:nested}.
\end{proof}

\section{Common contamination and a binary depth extremizer}

Let $\nu$ be any common background law and
\begin{equation}
 \widetilde\mu_a=(1-\varepsilon)\mu_a+\varepsilon\nu.
 \label{eq:common-noise}
\end{equation}
All sub-$d^\perp$ marginals remain identical, while for distinct cosets
\begin{equation}
 \frac12\|\mathfrak M_{\widetilde\mu_a}^{(n)}-
 \mathfrak M_{\widetilde\mu_b}^{(n)}\|_{\rm strat}
 =(1-\varepsilon).
 \label{eq:robust-distance}
\end{equation}
The common component cancels, and \cref{thm:isometry} gives the equality.

A minimal depth-extremizing implementation uses $G=\mathbb Z_2$ and its regular qubit representation $U_0=I$, $U_1=X$.  At each time, prepare $|\Phi^+\rangle$ between the process input and a retained qubit.  The two branches produce the orthogonal Bell states $|\Phi^+\rangle$ and $|\Psi^+\rangle$.  Bell measurement returns the history bit $g_j$, and the XOR of all $n$ outcomes reads the even/odd sector.  Every proper subset of outcomes is exactly uniform.

With independent Bell-readout error probability $\eta$, the parity bit is wrong with probability
\begin{equation}
 P_{\rm parity\ error}=\frac{1-(1-2\eta)^n}{2}.
 \label{eq:parity-error}
\end{equation}
This formula describes stochastic readout noise, but it is not an adversarial correction guarantee: all parity cosets are distance one apart.  Robust sector decoding is instead supplied by the nested construction in \cref{thm:nested}.

\begin{figure}[t]
\centering
\includegraphics[width=\columnwidth]{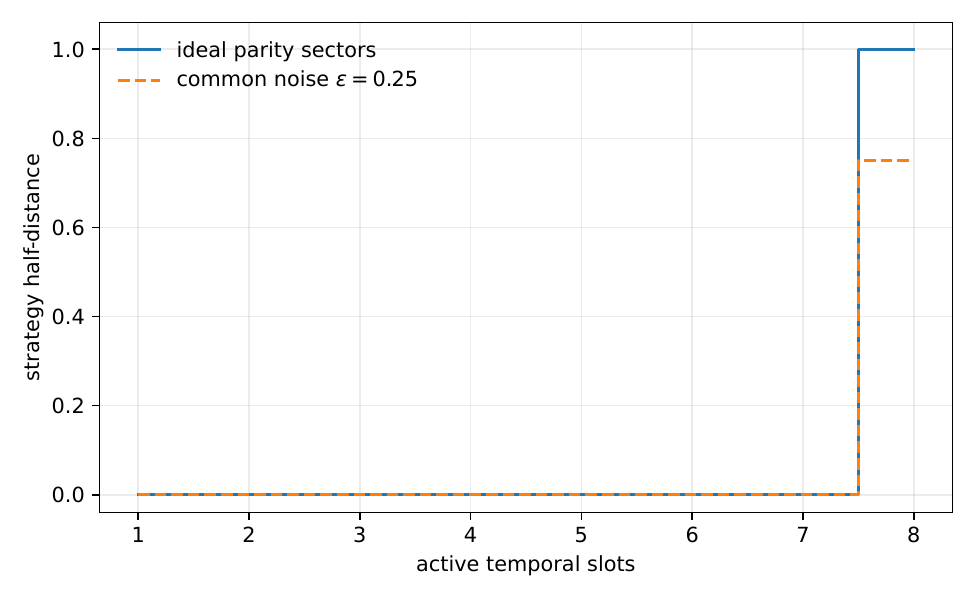}
\caption{Binary parity histories have zero operational distance at every proper depth and unit strategy half-distance at depth $n$.  Common background noise reduces only the terminal value.}
\label{fig:parity}
\end{figure}

\section{Memory cost and process tomography}

A code-coset trajectory has a finite classical dilation.  Sample $z\in\F_q^k$, output $Gz+a$, and retain the generator state needed by a sequential implementation.  The parity family is particularly economical: sample the first $n-1$ symbols locally and keep one running $q$-ary sum to choose the final symbol.  For the parity family, maximal observational depth can therefore coexist with a constant-size persistent environment register plus fresh randomness; a generic linear code may require a larger sequential realization.

The hierarchy exposes a precise failure of reduced process tomography.  For binary parity sectors, even the collection of all $(n-1)$-time marginals does not determine the full process.  This is stronger than the insufficiency of pair correlations.  Conversely, \cref{thm:approx-capacity} quantifies the largest ambiguity compatible with imperfect $t$-time privacy and imperfect global reconstruction.

The result should not be confused with quantum Markov order.  Generator memory dimension, Markov order, tester depth, and hidden-sector capacity are distinct resources.  Our bounds concern the last two.

\section{Exact boundary between coherent and classical histories}
\label{sec:coherent-boundary}

\begin{proposition}[Orthogonal records dephase a coherent history control]
\label{prop:coherent-history-dephasing}
Let $C$ have orthonormal basis $\{|\bm g\rangle:\bm g\in G^n\}$ and initial
state
$\rho_C=\sum_{\bm g,\bm h}c_{\bm g\bm h}|\bm g\rangle\langle\bm h|$.
Condition the $n$ process slots coherently on the history label.  If the
control is traced out, or if an inaccessible environment stores mutually orthogonal record
states $|e_{\bm g}\rangle$ and is traced out, then the reduced process is exactly
\begin{equation}
 \sum_{\bm g}c_{\bm g\bm g}
 \bigotimes_{j=1}^n\operatorname{Ad}_{U_{g_j}}
 =\mathfrak M_\mu^{(n)},
 \qquad \mu(\bm g)=c_{\bm g\bm g}.
 \label{eq:coherent-to-classical-history}
\end{equation}
Thus every theorem of this paper applies without modification to a coherent
implementation once an orthogonal which-history record is discarded.  The
off-diagonal phases are operationally accessible only if the control is
retained and interfered, or if the record states have nonzero overlaps.
\end{proposition}

\begin{proof}
Before tracing, the controlled evolution contains cross terms
$c_{\bm g\bm h}|\bm g\rangle\langle\bm h|$ tensored with the corresponding
left and right history branches.  The partial trace over $C$ multiplies each
cross term by $\langle\bm h|\bm g\rangle=\delta_{\bm g\bm h}$.  An orthogonal
environmental record gives the identical factor
$\langle e_{\bm h}|e_{\bm g}\rangle$.  Only the diagonal probabilities
remain, yielding \eqref{eq:coherent-to-classical-history}.  Nonorthogonal
records or an interferometric measurement of $C$ retain cross terms, so the
classical trajectory coordinate is then incomplete.
\end{proof}

The following is the quantum Singleton argument expressed in temporal-slot language \cite{KnillLaflamme2000,Grassl2022}.

\begin{theorem}[Coherent-history Singleton bound]
\label{thm:coherent-singleton}
Let $V:\mathbb C^K\to(\mathbb C^q)^{\otimes n}$ be an isometric coherent
history encoder with $K>1$.  Suppose that for every set $E$ of at most $t$ temporal
registers there is a state $\omega_E$ such that
\begin{equation}
 \operatorname{Tr}_{E^c}
 V|\psi\rangle\langle\phi|V^\dagger
 =\langle\phi|\psi\rangle\,\omega_E
 \label{eq:phase-sensitive-privacy}
\end{equation}
for all logical vectors $|\psi\rangle,|\phi\rangle$.  Thus every $t$-time
observer is blind not only to classical labels but also to all logical
phases.  Then necessarily $2t<n$, and
\begin{equation}
 \boxed{\log_qK+2t\leq n.}
 \label{eq:quantum-singleton-history}
\end{equation}
Condition \eqref{eq:phase-sensitive-privacy} is equivalent to exact
correctability of erasure of any $t$ temporal registers.  Quantum MDS codes,
when they exist for the chosen parameters, attain equality.  By contrast,
the diagonal classical history families of the preceding sections can carry
$q^{n-t}$ labels; coherent phase privacy therefore costs a second block of
$t$ registers.
\end{theorem}

\begin{proof}
Equation~\eqref{eq:phase-sensitive-privacy} is the Knill--Laflamme condition
for erasure of $E$: all operators supported on $E$ act as scalars after
compression to the code.  If $2t\geq n$, two disjoint surviving subsystems
of size at most $n-t$ would each support exact recovery, contradicting
no-cloning for $K>1$.  Hence $2t<n$.  Purify the maximally mixed logical
state with a reference $R$, and partition the physical registers into
disjoint sets $A,B,C$ with $|A|=|B|=t$.  Privacy gives
$I(R:A)=I(R:B)=0$.  Since the joint state on $RABC$ is pure,
\[
 S(BC)=S(R)+S(A),\qquad S(AC)=S(R)+S(B).
\]
Subadditivity for the two left sides yields
$2S(R)\leq2S(C)$.  Hence
\[
 \log K=S(R)\leq S(C)\leq(n-2t)\log q,
\]
which is \eqref{eq:quantum-singleton-history}.
\end{proof}

\section{Limitations and extensions}

The exact total-variation and convolution theorems concern classical mixtures of unitary histories in the fresh-probe slot architecture defined above.  Theorem~\ref{thm:coherent-singleton} covers coherent code subspaces under phase-sensitive privacy; arbitrary quantum-controlled histories and indefinite causal order remain outside the present classification.  Any reduction of broader dynamics to stochastic trajectories discards coherence whose operational significance must be analyzed separately.

Exact orbit resolution requires complete irreducible support.  Incomplete carriers compress or identify trajectories, replacing total variation by a smaller channel norm.  Approximate designs and continuous groups require energy or bandwidth constraints; no finite-dimensional carrier supports an exactly orthogonal continuum of branches.

The causal degradation theorem is intentionally output-side.  Superprocesses that modify future input interfaces, merge time steps, change the process arity, or coherently propagate one probe through several branches form a larger resource theory and can evade the convolution description.

\section{Exact coherent-hiding criterion}
\label{sec:coherent-hiding-criterion}

Classical marginal privacy and coherent quantum privacy are different
requirements.  The precise boundary is the erasure-correction condition.

\begin{theorem}[Coherent temporal hiding equals erasure correction]
\label{thm:coherent-hiding-erasure}
Let $\mathcal Q\subseteq\bigotimes_{i=1}^n\mathcal H_i$ be a code subspace
with projector $P$, and let $S$ be a set of temporal slots.  The following are
equivalent:
\begin{enumerate}[label=(\roman*)]
\item the reduced state on $S$ is the same for every density operator
supported on $\mathcal Q$;
\item for every operator $A_S$ supported on $S$ there is a scalar
$c(A_S)$ such that
\[
 P(A_S\otimes I_{S^c})P=c(A_S)P;
\]
\item erasure of the slots $S$ is exactly correctable on $\mathcal Q$.
\end{enumerate}
Hence a family hides arbitrary coherent logical superpositions from every
set of at most $t$ slots if and only if it is a quantum code correcting every
$t$-slot erasure.
\end{theorem}

\begin{proof}
If all reduced states on $S$ coincide, the expectation of every $A_S$ is
constant on unit vectors in $\mathcal Q$.  Polarization then gives
$P(A_S\otimes I)P=c(A_S)P$.  The converse follows by taking expectations.
Condition (ii) is the Knill--Laflamme condition \cite{KnillLaflamme2000} for the operator system spanned by errors supported on the erased slots, and is equivalent to exact recovery of that erasure.
\end{proof}

\begin{corollary}[Boundary of the classical code-coset construction]
Classical $t$-wise independent trajectory labels guarantee only equality of
diagonal temporal marginals.  They extend to coherent $t$-slot hiding exactly
when their coherent span satisfies Theorem~\ref{thm:coherent-hiding-erasure};
equivalently, the associated quantum code has distance at least $t+1$.
\end{corollary}

\begin{proof}
Equality of diagonal marginals controls only diagonal observables.  By \Cref{thm:coherent-hiding-erasure}, equality for all coherent superpositions is equivalent to the Knill--Laflamme condition for every $t$-slot erasure, which is precisely distance at least $t+1$.
\end{proof}

\section*{orbit-frame error decomposition}

Imperfect orbit resolution and coherent process perturbations can be
separated by combining frame conditioning with the strategy norm.

\begin{theorem}[Resolved-frame hiding stability]
\label{thm:third-hiding-frame}
Let \(\mathcal P\) be an ideal trajectory-mixture process and
\(\widetilde{\mathcal P}\) a coherent implementation with
\[
 \|\widetilde{\mathcal P}-\mathcal P\|_{\rm strat}\le\varepsilon.
\]
Suppose orbit reconstruction uses a frame operator \(S\) with supported
lower bound \(\sigma>0\), and its implementation perturbation \(E\)
satisfies \(\|E\|<\sigma\).  If the remaining synthesis map has norm at
most \(C\), define
\[
 \delta_{\rm fr}
 =
 \frac{C\|E\|}{\sigma(\sigma-\|E\|)}.
\]
Then every adaptive tester sees total deviation at most
\[
 \varepsilon+\delta_{\rm fr}.
\]
Consequently an ideal hiding margin \(m\) remains positive whenever
\(\varepsilon+\delta_{\rm fr}<m\).
\end{theorem}

\begin{proof}
The frame inverse perturbation bound controls the reconstructed ideal
process by \(\delta_{\rm fr}\).  Add the coherent process error and use
the operational definition of the strategy norm, which already optimizes
over adaptive testers.
\end{proof}

This identifies two different phase boundaries.  The process may leave
the classical trajectory neighbourhood as \(\varepsilon\) grows, or the
orbit frame may become singular as \(\sigma\downarrow0\).  A claimed
coherent advantage must remain after both errors have been charged; it
cannot be attributed to ill-conditioned orbit reconstruction.

\section*{an exact coherent-history separation}

A coherent process sector outside the classical trajectory polytope can
already be realized by a replacer comb.

\begin{theorem}[Half-unit separation from every history mixture]
\label{thm:frontier-hiding-coherent}
Let \(h_0,h_1\) be two orthogonal history records and let a process ignore
its inputs and output
\[
 |+\rangle=\frac{|h_0\rangle+|h_1\rangle}{\sqrt2}.
\]
For every classical trajectory-mixture process whose record state is
diagonal in the history basis,
\[
 \frac12\|\mathfrak P_+-\mathfrak P_{\rm mix}\|_{\rm strat}
 \ge\frac12.
\]
Equality is attained by the equal mixture of \(h_0\) and \(h_1\).
\end{theorem}

\begin{proof}
For replacer combs the strategy norm is the trace norm of the output-state
difference.  Dephasing in the two-history subspace fixes every classical
mixture and maps \(|+\rangle\langle+|\) to the equal diagonal mixture.
The off-diagonal witness
\(|h_0\rangle\langle h_1|+|h_1\rangle\langle h_0|\) has expectation one
on \(|+\rangle\) and zero on every diagonal state, giving trace distance
at least \(1/2\).  Direct diagonalization against the equal mixture gives
eigenvalues \(\pm1/2\), so equality holds.
\end{proof}

This constructs the coherent sector demanded by the claim boundary and
proves it cannot be explained by unresolved classical histories.  Coding
rates and adaptive strong converses remain separate asymptotic questions.

\section*{Exact separation for an arbitrary history alphabet}

The half-unit binary gap is the first member of a closed \(n\)-history
law.

\begin{theorem}[Uniform coherent-history distance]
\label{thm:final-hiding-nhistory}
Let \(h_1,\ldots,h_n\) be orthogonal history records and let a replacer
comb output
\[
 |u_n\rangle=\frac1{\sqrt n}\sum_{j=1}^n|h_j\rangle .
\]
Then its minimum strategy distance from every classical history-mixture
replacer is
\[
 \min_{\substack{q_j\ge0\\\sum_jq_j=1}}
 \frac12\left\|
 |u_n\rangle\langle u_n|-\sum_jq_j|h_j\rangle\langle h_j|
 \right\|_1
 =1-\frac1n.
\]
The unique permutation-symmetric minimizer is the uniform classical
mixture.
\end{theorem}

\begin{proof}
The coherent state is invariant under permutations of the history basis.
Average an arbitrary diagonal competitor over all permutations.  Convexity
and unitary invariance of trace distance show that this cannot increase
the distance and produces \(I_n/n\).  The difference
\(|u_n\rangle\langle u_n|-I_n/n\) has eigenvalue \(1-1/n\) on
\(|u_n\rangle\) and eigenvalue \(-1/n\) with multiplicity \(n-1\).
Half its trace norm is \(1-1/n\).
\end{proof}

Coherent temporal information therefore becomes asymptotically unit
separated from every unresolved classical trajectory model as the history
alphabet grows.

\begin{corollary}[Exact adaptive strong converse for repeated coherent sectors]
\label{cor:final-hiding-strong-converse}
For \(m\) independent uses of the uniform \(n\)-history coherent replacer,
allow the competing classical model to use an arbitrary correlated
probability law on the \(n^m\) product histories.  Then
\[
 \min_{\mathfrak P_{\rm cl}}
 \frac12\left\|
   \mathfrak P_{u_n}^{\otimes m}-\mathfrak P_{\rm cl}
 \right\|_{\rm strat}
 =1-n^{-m}.
\]
Consequently the optimal equal-prior discrimination error between the
coherent process and its closest classical-history process is
\[
 p_{\rm err}^{(m)}=\frac{1}{2n^m}.
\]
The statement already optimizes over adaptive process testers.
\end{corollary}

\begin{proof}
The repeated replacer outputs the uniform coherent state on the
\(n^m\)-element product-history basis.  A general correlated classical
competitor is exactly an arbitrary state diagonal in that basis.
Theorem~\ref{thm:final-hiding-nhistory}, applied with alphabet size
\(n^m\), gives the distance \(1-n^{-m}\).  Since the processes ignore
every input, their strategy norm is the trace norm of the output-state
difference; adaptive choices cannot alter it.  The Helstrom identity
\(p_{\rm err}=(1-D)/2\) gives the second formula.
\end{proof}

Thus this coherent replacer family has a closed strong converse with
exponent \(\log n\), even against correlated classical histories and
fully adaptive discrimination.  This does not assert the same exponent
for general memory-bearing coherent combs.

\section{Conclusion}

Correlated finite-frame histories admit an exact operational coordinate
whenever the carrier resolves the group orbit.  Strategy distance is
total variation, and output-side causal degradation is classical
convolution.  Code cosets then produce prescribed temporal-depth
separations, while an entropy converse remains valid with leakage and
decoding error.  Even for nonlinear, nonuniform sectors, a family of
$M$ labels hidden on $t$ times and separated by distance $d_Z$ obeys
$\log_qM+t+d_Z-1\le n$.  Robust linear readout is governed by nested,
not unrestricted, cosets; nested MDS pairs attain the resulting
$R+t+2e+f\le n$ boundary.  These results turn the qualitative
incompleteness of reduced temporal channels into finite-error and
adversarial coding theorems for quantum reference processes.

\appendix

\section{Comb normalization and the trajectory upper bound}
\label{app:comb}

For completeness, the notation in \cref{eq:history-process} can be embedded in the standard comb formalism as follows.  Let $\mathcal G_{\bm g}$ denote the deterministic $n$-slot comb that inserts $\operatorname{Ad}_{U_{g_j}}$ in slot $j$.  The process is the convex combination
\begin{equation}
 \mathfrak M_\mu^{(n)}=\sum_{\bm g}\mu(\bm g)\mathcal G_{\bm g}.
\end{equation}
A normalized strategy tester $T$ maps every deterministic comb to a density operator before the terminal measurement.  Consequently
\begin{equation}
 T[\mathfrak M_\mu^{(n)}-\mathfrak M_\nu^{(n)}]
 =\sum_{\bm g}[\mu(\bm g)-\nu(\bm g)]\rho_{\bm g},
\end{equation}
where $\rho_{\bm g}=T[\mathcal G_{\bm g}]$ has unit trace.  This justifies the triangle-inequality upper bound in \cref{thm:isometry} without choosing a particular Choi normalization.  The lower tester is a legitimate parallel comb tester: it prepares the product orbit seed, connects one carrier factor to each slot, retains the reference factors, and performs the orthogonal history measurement only after the last slot.  No temporal feed-forward is used.

\section{Adaptive coordinate queries}

Suppose two history laws have identical marginals on every coordinate set of size at most $t$.  Consider a classical decision tree that queries at most $t$ coordinates, choosing each new coordinate from the previous transcript.  Along a fixed leaf, the event defining that transcript specifies values on a set of at most $t$ coordinates, so it has the same probability under both laws.  Every transcript distribution is therefore identical.  In the orbit-resolving process, measuring branch labels before adaptive control reduces any branch-query protocol to this argument.  More general coherent testers supported on a fixed set of at most $t$ slots are already covered by equality of the reduced combs.  This statement does not cover a coherent superposition of different temporal interfaces whose union has support larger than $t$.

\section{Why all cosets cannot correct one adversarial error}
\label{app:coset-obstruction}

Let $C<\F_q^n$ be proper and use every coset in $\F_q^n/C$ as a sector.  There is a standard basis vector $e_i\notin C$; otherwise $C$ would contain a basis and equal the whole space.  For every history $x\in a+C$, the one-symbol error $x\mapsto x+\alpha e_i$ with $\alpha\ne0$ produces a valid history in the distinct sector $a+\alpha e_i+C$.  An observation equal to that corrupted history is therefore compatible both with a noiseless preparation of the second sector and with a one-error preparation of the first.  No decoder can recover the sector against all one-symbol errors.  This elementary obstruction is the reason that the distance $d(C)$ cannot be used as the sector distance when all cosets are labels.

\section{Nested generalized Reed--Solomon extremizers}
\label{app:nested-grs}

Choose distinct evaluation points $\alpha_1,\ldots,\alpha_n\in\F_q$ and nonzero multipliers $v_i$.  For $1\le j\le n$, write
\begin{equation}
 \operatorname{GRS}_j(\bm\alpha,\bm v)
 =\{(v_1f(\alpha_1),\ldots,v_nf(\alpha_n)):\deg f<j\}.
\end{equation}
For $k<K$, set $C=\operatorname{GRS}_k$ and $D=\operatorname{GRS}_K$.  The nesting is immediate.  Every nonzero word of $D$ has weight at least $n-K+1$, and a polynomial of degree $K-1$ with $K-1$ prescribed roots gives a word of that weight.  Such a word cannot lie in $C$, whose distance is $n-k+1$, so
\begin{equation}
 \drel(D,C)=n-K+1.
\end{equation}
The duals are generalized Reed--Solomon codes of dimensions $n-k$ and $n-K$.  The same argument gives
\begin{equation}
 \drel(C^\perp,D^\perp)=k+1.
\end{equation}
This proves \cref{eq:mds-nested-parameters} directly.  A decoder may first decode the received word to the unique word of $D$ within radius $\lfloor(n-K)/2\rfloor$ and then output its quotient class modulo $C$.

\section{Extension beyond linear and uniform sectors}
\label{app:nonlinear}

The entropy converse \cref{thm:approx-capacity} does not require linearity, uniformity within a sector, or pairwise disjoint supports.  These hypotheses enter only the explicit code constructions.  Conversely, the exact strategy isometry means that any classical family of trajectory laws---including nonlinear orthogonal arrays, almost-universal hash families, or hidden Markov laws---immediately defines a process family with the same total-variation geometry.  What does not transfer automatically is causal degradation under transformations that modify input ports: \cref{thm:causal} remains restricted to output-side converters.

\section{Reproducibility}

The ancillary verification script first checks the orbit-resolution character identity and the non-Abelian convolution kernel on the group $S_3$, and then enumerates binary and ternary parity
sectors, verifies all proper marginals
exactly, and constructs ordinary and nested Reed--Solomon examples over
prime fields.  For the $[6,2]_7\subset[6,4]_7$ pair it enumerates all
$49$ quotient sectors, their projections, relative distance, and
radius-one error balls.  It checks the exact and finite-error capacity
inequalities on $250$ random joint laws, verifies the nonlinear and
linear payload--depth--protection bounds, and regenerates all three
figures.  Discrete counts are performed with integers before
conversion to floating point.

\section*{Data and code availability}
All scripts and machine-readable verification records needed to reproduce
the finite examples and figures are supplied in the ancillary directory of
the arXiv source package.  No proprietary or access-restricted data were used.

\section*{Scope, limitations, and open problems}

The paper identifies an exact total-variation strategy metric when the
carrier resolves the finite-group orbit, constructs code-based temporal
hiding sectors, and proves exact and approximate entropy converses.  The
verifier audits the non-Abelian \(S_3\) kernel, parity sectors, nested
Reed--Solomon example, capacity inequalities, and all figures.  Output-side
convolution is the proved causal degradation; input-port transformations are
outside that theorem.

The subsequent theorem now obtains the continuous compact-group limit for
smooth classical history laws with a finite energy cutoff and an explicit
Sobolev error.  The remaining target is therefore cleanly quantum: construct
coherent process sectors that are not trajectory mixtures, derive the
appropriate comb or strategy norm after imperfect orbit resolution, and
prove a strong converse for adaptive decoders.

\section*{coherent-perturbation bound}

Let \(\Upsilon_0,\Upsilon_1\) be two coherent process sectors and let
\(\Gamma_0,\Gamma_1\) be classical-trajectory mixtures satisfying
\(\|\Upsilon_j-\Gamma_j\|_{\mathrm{strat}}\le\varepsilon_j\) in the
strategy norm.  The triangle inequality gives
\[
 \left|
 \|\Upsilon_0-\Upsilon_1\|_{\mathrm{strat}}
 -\|\Gamma_0-\Gamma_1\|_{\mathrm{strat}}
 \right|
 \le\varepsilon_0+\varepsilon_1.
\]
Because the strategy norm already optimizes over admissible adaptive
testers, the estimate includes adaptive decoders without an additional
argument.

Hence every classical hiding construction with margin \(m\) remains secure
against coherent perturbations of total size below \(m\), and every
classical distinguishability witness remains effective up to the same
additive loss.  Imperfect orbit resolution can be incorporated by adding
its channel distance through data processing.

This provides a rigorous bridge from the exact trajectory theory to a
neighbourhood of genuinely coherent processes.  A strict quantum
separation still requires a sector outside every such neighbourhood,
together with a strong converse not inherited from the classical mixture.

\section*{compact-group history resolution}

The compact-group band estimate transfers to an entire finite history by
working on the product group rather than treating the time slots
independently.

\begin{theorem}[Smooth-history anti-hiding bound]
\label{thm:fourth-smooth-history}
Let \(K\) be a compact Lie group, let the history space be \(K^T\), and
let \(\mathfrak P_\mu\) denote the classical trajectory-mixture comb with
history density \(\mu\).  Let \(\Pi_E^{(T)}\) be the spectral projector of
the product Laplacian on \(K^T\).  Assume the retained history
representations are resolved by the carrier and define the positive
band constant
\[
 c_{E,T}=
 \inf_{\substack{\sigma\in\operatorname{ran}\Pi_E^{(T)}\\
                  \|\sigma\|_1=1}}
 \|\mathfrak P_\sigma\|_{\rm strat}.
\]
For history laws \(p,q\) with
\(\sigma=p-q\in H^s(K^T)\), set
\[
 D_{\rm strat}(p,q)
 =\frac12\|\mathfrak P_p-\mathfrak P_q\|_{\rm strat}.
\]
Then
\[
 D_{\rm strat}(p,q)
 \ge
 c_{E,T}\TV(p,q)
 -\frac{1+c_{E,T}}2(1+E)^{-s/2}
 \|\sigma\|_{H^s}.
\]
\end{theorem}

\begin{proof}
On normalized Haar measure,
\[
 \|(I-\Pi_E^{(T)})\sigma\|_1
 \le(1+E)^{-s/2}\|\sigma\|_{H^s}.
\]
The definition of \(c_{E,T}\) controls the retained part, while the
strategy norm of the tail is at most its \(L^1\)-norm because a
trajectory mixture is a convex linear combination of normalized combs.
The same retained-plus-tail triangle argument as for the one-time
diamond norm gives the result.
\end{proof}

Thus a smooth classical payload of total-variation size \(\tau\) cannot
be hidden below
\[
 c_{E,T}\tau
 -\frac{1+c_{E,T}}2(1+E)^{-s/2}
 \|p-q\|_{H^s}.
\]
This completes the finite-energy compact-group limit for classical
trajectory mixtures.  It also isolates the genuinely quantum target:
a coherent-history separation must violate the mixture representation
itself, rather than exploit an uncharged high-frequency tail or an
unresolved orbit frame.

\makeatletter
\let\auto@bib@innerbib\relax
\let\auto@bib\relax
\let\write@bibliographystyle\relax
\makeatother

\end{document}